\documentclass[11pt]{article}
\usepackage{latexsym,amsmath,amsthm,amssymb}
\usepackage{mathrsfs,wasysym,graphicx,lscape}
\usepackage{esint}
\usepackage{color}
\begin{document}

\newtheorem{theorem}{\indent Theorem}[section]
\newtheorem{proposition}[theorem]{\indent Proposition}
\newtheorem{definition}[theorem]{\indent Definition}
\newtheorem{lemma}[theorem]{\indent Lemma}
\newtheorem{remark}[theorem]{\indent Remark}
\newtheorem{corollary}[theorem]{\indent Corollary}

\begin{center}
    {\large \bf Bohmian Trajectories of the Time-oscillating Schr\"{o}dinger Equations}
\vspace{0.5cm}\\{\sc Dandan Li$^{*,1}$, Jinqiao Duan$^{2}$, Li Lin$^{1}$ and Ao Zhang$^{1}$}\\
{\small 1. School of Mathematics and Statistics and Center for Mathematical Sciences,\\ Huazhong University of Science and Technology,
 Wuhan, 430074, PR China}\\
{\small 2. Departments of Applied Mathematics \& Physics, Illinois Institute of\\ Technology,
Chicago, Illinois, 60616, USA}\\

\end{center}

\renewcommand{\theequation}{\arabic{section}.\arabic{equation}}
\numberwithin{equation}{section}

\begin{abstract}
Bohmian mechanics is a non-relativistic quantum theory based on a particle approach.
In this paper we study the Schr\"{o}dinger equation with rapidly oscillating potential and the associated Bohmian trajectory.
We prove that the corresponding Bohmian trajectory converges locally in measure, and the limit coincides with the Bohmian trajectory for the effective Schr\"{o}dinger equation on a finite time interval. This is beneficial for the efficient simulation of the Bohmian trajectories in oscillating potential fields.
\end{abstract}


\footnote[0]{\hspace*{-7.4mm}
$^*$Corresponding author\\
Emails: dandanli@hust.edu.cn (D. Li); duan@iit.edu (J. Duan); linli@hust.edu.cn (L. Lin);\\
\quad\quad zhangao1993@hust.edu.cn (A. Zhang)
}
\section*{Lead Paragraph}
Bohmian mechanics is a quantum theory of motion { for} concerning particles. Each particle has a position at all times and the evolution of these positions is governed by the usual quantum wave function which { satisfies} the Schr\"{o}dinger equation. Roughly speaking, according to the theory of Bohmian mechanics the trajectories of electrons around atoms can be pictured as the motion of the solar system. The Schr\"{o}dinger equation with rapidly oscillating potential { arises} in many fields. { It has been shown} that the large scale Schr\"{o}dinger equation converges to an effective Schr\"{o}dinger equation with averaged or homogenized potential under certain assumptions. A natural question is whether the corresponding Bohmian trajectory converges to the one of the effective Schr\"{o}dinger equation. Due to the complicated relation between wave functions and the Bohmian trajectories, the convergence of the Schr\"{o}dinger equations does not directly imply convergence under the corresponding Bohmian trajectories. The convergence of Bohmian measure can be obtained from convergence of wave function, and then combined with the properties of Young measure corresponding to Bohmian flow, the convergence of Bohmian trajectory in measure is obtained.

\section{Introduction}
\noindent

Bohmian mechanics is a Galilei-invariant theory for the motion of particles { (\cite{1952Bohm})}. The theory was developed by  Bohm in 1952  as an alternative explanation of the motion for particles in quantum mechanics.
Mathematically, the law of motion is a first order differential equation and thus Bohmian mechanics is { naturally formulated in configuration space}. The configuration space of $N$ particles is the vector
\begin{equation*}
   {q=(q_1,\cdots,q_N)\in\mathbb{R}^{3N}},
\end{equation*}
where $ {q_j\in\mathbb{R}^3}$ is the position of { $j$-th} particle.
Given an initial state, the theory yields a trajectory
\begin{equation*}
   {t\to X(t)\in\mathbb{R}^{3N}},
\end{equation*}
in configuration space, i.e. the position of particles at all time.

The positions of the particles can be obtained from Bohmian mechanics, and then velocities, momentum, energies and spin can be found through the positions { (\cite{2014Benseny,2014sanz,2009Teufel})}. The law of motion for ${X(t)\in\mathbb{R}^N}$ is
{ \begin{equation}\label{in1}
  {\frac{\mathrm{d}}{\mathrm{d}t}X(t)=\frac{\mathrm{Im}\left(\overline{\psi}(t,X(t))\nabla\psi(t,X(t))\right)}{|\psi(t,X(t))|^2}},
\end{equation}}
where the wave function $ {\psi(t,x):\mathbb{R}\times\mathbb{R}^{3N}\to\mathbb{C}}$
solves the Schr\"{o}dinger equation { with} a Hamiltonian operator $H$
\begin{equation*}
   {i\frac{\partial}{\partial t}\psi(t,x)=H \psi(t,x)}.
\end{equation*}
The wave function guarantees a time-dependent vector-field on configuration space. Its integral curves are the possible trajectories of Bohmian particles. Given initial data $ {X(0)}$ and $ {\psi(0)}$, the position ${X(t)}$ is determined for all time. The mathematical
foundation of \eqref{in1} is solid and the well-posedness (in the sense of some measure) can be found in { \cite{2005Teufel}} regardless of the continuity of the velocity field.

In this paper, we study the Bohmian trajectory corresponding to the large scale Schr\"{o}dinger equation under rapidly oscillating potential. In recent years, the Schr\"{o}dinger equation with  rapidly oscillating potential function and rapidly oscillating force function has been extensively studied as a fundamental equation in modern mathematical physics. The large scale Schr\"{o}dinger equation can be derived from plasma physics and as an amplitude equation in a perturbation study of sine-Gordon equation (\cite{1990Bishop}).  The problem that involved with periodic coefficient was investigated in \cite{1978Lions} and time-dependent periodic potential with large scale has been studied intensely in \cite{2005Allaire,2014Signing}.  Meanwhile, since the importance of incorporating stochastic effects in the modeling of complex systems has been recognized, the oscillating stochastic Schr\"{o}dinger equation had also attracted { a lot of attention} (\cite{2014duan,2019duan}). It has been { shown} that the large scale Schr\"{o}dinger equation converges to an effective Schr\"{o}dinger equation with averaged or homogenized potential under certain assumptions. A natural question is whether the corresponding Bohmian trajectory converges to the one of the effective Schr\"{o}dinger equation.

The study of the Bohmian trajectories arises from the semiclassically scaled Schr\"{o}dinger equations { (\cite{2016Sparber})}.
For the scaled Schr\"{o}dinger equation, the limit of the wave functions and the corresponding Bohmian trajectories had been studied intensely. Note that nothing can yield by simply passing the scale parameter to 0, the classical limit of the scaled Schr\"{o}dinger equation falls into the so called { ``singular limits"} in asymptotic approximation theory (\cite{2009Fe,2010Majda,2021Dan}) which has been studied in a variety of differential equations. However, as for {  the convergence} of Bohmian trajectories, since the relation between the corresponding Bohmian trajectories and wave functions is quite complicated, it is very difficult to prove the convergence of Bohmian trajectories directly through the results of wave functions.
In recent years, the study of the classical limit of Bohmian trajectories is of great interest. For wave packets
as defined by Hagedorn, {  it was shown in \cite{2010D}} that the Bohmian trajectories converge to Newtonian trajectories in probability.
Instead of studying Bohmian trajectories directly, \cite{2010Markowich}  considered a class of phase space measure, i.e. Bohmian measure, which naturally arises in the Bohmian interpretation of quantum mechanics.  The authors of \cite{2014Figalli} assumed that the initial data for the scaled Schr\"{o}dinger equation were { WKB (after three papers by Wentzel,
Kramers and Brillouin respectively, in 1926)} and used the theory of Young measure as a connection between the Bohmian flow and its limit. The analysis of oscillation and concentration effects in  \cite{2014Figalli} and \cite{2010Markowich} was in the semi-classical regime.

Thus, the results in \cite{2010D,2014Figalli,2010Markowich} and the study of convergence of the Schr\"{o}dinger equation with rapidly oscillating functions provide us with a  possibility to prove the convergence of corresponding Bohmian trajectory.
 In this paper, the convergence of the wave functions governed by the time-oscillating Schr\"{o}dinger equations is  proved under certain assumptions. { This is then to consider the  limit} of the corresponding Bohmian trajectories. We start by deriving the limit of Bohmian measures in dependence on the scale of oscilations and concentrations of the sequence of wave functions under consideration. { Furthermore}, given the equivalence between Bohmian measure and the Bohmian flow, { we also prove rigorously that the corresponding Bohmian trajectory converges locally in measure} on the finite time interval, and the limit trajectory coincides with the effective system of the time-oscillating Schr\"{o}dinger equation.

The rest of the paper is organized as follows. Section 2 is devoted to the basic setting of our problem and some preliminary results in oscillating periodic functions, measure theory, whereas in Section 3 the convergence theorems (Theorem \ref{conhs} and Theorem \ref{step3}) are established for both wave function satisfying the time-oscillating Schr\"{o}dinger equation and the associated  Bohmian trajectory. {  Finally, Section 4 is the Conclusion and Discussion section.}

\section{Preliminaries}
\noindent

Consider the  wave function $\psi^\varepsilon(t,\cdot)$ in $\mathbb{R}^N$  satisfying the following time-oscillating Schr\"{o}dinger equation:
\begin{equation} \label{hs}
\begin{cases}
i\partial_t\psi^\varepsilon(t,x) = -\frac{1}{2}\Delta \psi^\varepsilon(t,x) + V(\frac{t}{\varepsilon},x)\psi^\varepsilon(t,x),\quad x\in\mathbb{R}^N,~ t>0,\\
\psi^\varepsilon(t,x)|_{t=0}=\psi_0^\varepsilon,
\end{cases}
\end{equation}
where $\psi_0^\varepsilon\in H^1(\mathbb{R}^N)$, $N\geq 3$, $\varepsilon>0$ and { the Hamiltonian operator $H=-\frac{1}{2}\Delta + V(\cdot,\cdot)$}.

The potential $V(t,x)\in C_0^\infty(\mathbb{R}\times\mathbb{R}^N;\mathbb{R})$ with period 1 is assumed to be bounded below and subquadratic, i.e.,
\begin{equation}\label{Vass}
  \partial_x^k V(\cdot,x)\in L^\infty(\mathbb{R}^N) \quad \forall~k\in\mathbb{N}^n~\text{such that } |k|\geq2.
\end{equation}
Let $V^\star(x):=\int_0^1V(t,x)\mathrm{d}t$, the effective system of the time-oscillating Schr\"{o}dinger equation \eqref{hs} is written as
\begin{equation} \label{ehs}
\begin{cases}
i\partial_t\psi(t,x) = -\frac{1}{2}\Delta \psi(t,x) + V^\star(x)\psi(t,x),\quad x\in\mathbb{R}^N,~ t>0,\\
\psi(t,x)|_{t=0}=\psi_0,
\end{cases}
\end{equation}
where we assume that $\psi_0^\varepsilon\to\psi_0$ as $\varepsilon\to0$ strongly in $H^1(\mathbb{R}^N)$.

For the wave function $\psi^\varepsilon(t,x)$ one can associate two basic real valued densities, namely, the position density and the current density defined by
\begin{equation*}
  \rho^\varepsilon(t,x)=|\psi^\varepsilon(t,x)|^2,\quad J^\varepsilon(t,x)=\mathrm{Im}(\overline{\psi^\varepsilon}(t,x)\nabla\psi^\varepsilon(t,x)).
\end{equation*}
These quantities satisfy the conversation law
\begin{equation*}
  \partial_t\rho^\varepsilon + \mathrm{div}_xJ^\varepsilon =0.
\end{equation*}
The corresponding position density and current density of the {  effective wave function} $\psi(t,x)$ defined by
\begin{equation*}
  \rho(t,x)=|\psi(t,x)|^2,\quad J(t,x)=\mathrm{Im}(\overline{\psi}(t,x)\nabla\psi(t,x)).
\end{equation*}

Due to the Bohmian mechanics offered by Bohm in 1952 as an another approach to quantum mechanics,
one can define particle trajectories $X_t^\varepsilon:x\to X^\varepsilon(t,x)\in\mathbb{R}^N$ via the following differential equation
\begin{equation} \label{xhs}
\begin{cases}
\dot{X}^\varepsilon(t,x) = u^\varepsilon(t,X^\varepsilon(t,x)),\\
X^\varepsilon(t,x)|_{t=0}=x,\quad x\in\mathbb{R}^N,
\end{cases}
\end{equation}
where the velocity field is given by
\begin{equation*}
  u^\varepsilon(t,x):=\frac{J^\varepsilon(t,x)}{\rho^\varepsilon(t,x)}.
\end{equation*}
The particle trajectory $X_t:x\to X(t,x)\in\mathbb{R}^N$ corresponding to the {  effective wave function} $\psi(t,x)$ can be defined by
\begin{equation} \label{exhs}
\begin{cases}
{ \dot{X}(t,x) = u(t,X(t,x)),}\\
X(t,x)|_{t=0}=x,\quad x\in\mathbb{R}^N,
\end{cases}
\end{equation}
with the velocity field
\begin{equation*}
  u(t,x):=\frac{J(t,x)}{\rho(t,x)}.
\end{equation*}

To characterize the Bohmian mechanics more rigorously, we recall the Bohmian measure from \cite{2014Figalli} and \cite{2010Markowich} .
\begin{definition}[Bohmian measure]
Let $\varepsilon>0$ and $\psi^\varepsilon\in H^1(\mathbb{R}^N)$ be a sequence of wave functions with corresponding densities $\rho^\varepsilon$ and $J^\varepsilon$. Then the associated Bohmian measure $\beta^\varepsilon\in\mathcal{M}^+(\mathbb{R}_x^N\times\mathbb{R}_p^N)$ is given by
\begin{equation}\label{dbm}
  \langle\beta^\varepsilon,\varphi\rangle :=\int_{\mathbb{R}^N}\rho^\varepsilon(x)\varphi\left(x,\frac{\rho^\varepsilon(x)}{J^\varepsilon(x)}\right),\quad \forall~ \varphi\in C_0(\mathbb{R}_x^N\times\mathbb{R}_p^N),
\end{equation}
where $\mathcal{M}^+$ denotes the set of nonnegative Borel measures on phase space, $\langle\cdot,\cdot\rangle$ denotes the corresponding duality bracket between $\mathcal{M}^+(\mathbb{R}_x^N\times\mathbb{R}_p^N)$ and $C_0(\mathbb{R}_x^N\times\mathbb{R}_p^N)$, and $C_0$ is the closure (with respect to the uniform norm) of the set of continuous functions with compact support.
\end{definition}

We also recall the following assertions in \cite{2010Markowich} which ensures existence of a classical limit of $\beta^\varepsilon$.

\begin{lemma}[\cite{2010Markowich}]\label{wkb1}Let $\psi^\varepsilon$ be uniformly bounded in $L^2(\mathbb{R}^N)$. Then, up to extraction of sub-sequences, there exists a limiting measure $\beta\in\mathcal{M}^+(\mathbb{R}_x^N\times\mathbb{R}_p^N)$, such that
\begin{equation*}
  \beta^\varepsilon\to\beta \quad \text{as}\quad \varepsilon\to0_+\quad
\end{equation*}
weakly-$\star$ in $\mathcal{M}^+(\mathbb{R}_x^N\times\mathbb{R}_p^N)$.
\end{lemma}

\begin{theorem}[\cite{2010Markowich}]\label{wkb2}Let $\psi^\varepsilon$ be uniformly bounded in $H^1(\mathbb{R}^N)$ with corresponding densities $\rho^\varepsilon, J^\varepsilon\in L^1(\mathbb{R}^N)$. If $\rho^\varepsilon\xrightarrow{\varepsilon\to0_+}\rho$ in $L^1(\mathbb{R}^N)$ strongly and $J^\varepsilon\xrightarrow{\varepsilon\to0_+} \tilde{J}$ in measure, then $\beta$ is mono-kinetic, i.e.
\begin{equation*}
  \beta(x,p)=\rho(t,x)\delta\left(p-\frac{\tilde{J}(t,x)}{\rho(t,x)}\right).
\end{equation*}
\end{theorem}

Finally, we collect some basic fact from measure theory and homogenization theory.
\begin{theorem}[Weak limits of rapidly oscillating periodic functions \cite{1999Donato}]\label{wlr}
Let $1\leq p \leq +\infty$ and $f$ be a $Y$-periodic function in $L^p(Y)$. Set
\begin{equation*}
  f_\varepsilon(x)=f\left(\frac{x}{\varepsilon}\right)\quad a.e.~\text{on }\mathbb{R}^N.
\end{equation*}
Then, if $p<+\infty$, as $\varepsilon\to0$
\begin{equation*}
  f_\varepsilon\rightharpoonup \mathcal{M}_Y(f)=\frac{1}{|Y|}\int_Yf(y)\mathrm{d}y\quad \text{weakly in }L^p(\omega),
\end{equation*}
for every bounded open subset $\omega$ of $\mathbb{R}^N$.

If $p=+\infty$, one has
\begin{equation*}
  f_\varepsilon\rightharpoonup\mathcal{M}_Y(f)=\frac{1}{|Y|}\int_{Y}f(y)\mathrm{d}y\quad \text{weakly}-\star~\text{in }L^\infty(\mathbb{R}^N).
\end{equation*}
\end{theorem}

The following theorem shows the relationship between some sequence and the Young measure generated by it.
\begin{theorem}[\cite{1997Gasser}]\label{Ys}If $|\Omega|<\infty$ and $\nu_x$ is the Young measure generated by the (whole) sequence $\{u_j\}_{j\in\mathbb{N}}$ then
\begin{equation*}
  u_j\to u \text{ as }j\to\infty \text{ in measure }\Leftrightarrow \nu_x = \delta_{u(x)} \text{ for a.e. }x\in\Omega.
\end{equation*}
\end{theorem}

\section{Convergence of Bohmian Trajectories}
\noindent

In this section, we start with the existence and boundedness of solutions of the time-oscillating Schr\"{o}dinger equation \eqref{hs} and its effective system \eqref{ehs}. The results can be obtained by Galerkin approximation method and we omit the details here (\cite{2012Lions}).

\begin{lemma}\label{hsbdd}
Let the potential $V(t,x)\in C_0^\infty(\mathbb{R}\times\mathbb{R}^N;\mathbb{R})$ with period 1 is assumed to be bounded below and subquadratic, then there exists constant $T_0>0$ such that the wave function $\psi^\varepsilon(t,x)$ of \eqref{hs} is uniformly bounded in $H^1(\mathbb{R}^N)$ for all $t\in[0,T_0)$.
\end{lemma}

\begin{lemma}\label{ehsbdd}
Let the potential $V(t,x)\in C_0^\infty(\mathbb{R}\times\mathbb{R}^N;\mathbb{R})$ with period 1 is assumed to be bounded below and subquadratic, and  $V^\star(x)$ be as in Sect. 2, then there exists constant $T^\star_0>0$ such that he wave function $\psi(t,x)$ of \eqref{ehs} is uniformly bounded in $H^1(\mathbb{R}^N)$ for all $t\in[0,T^\star_0)$.
\end{lemma}

Note that { by use of} theorem of weak limits of rapidly oscillating periodic functions (Theorem \ref{wlr}) presented in \cite{1999Donato}, we obtain the following theorem for the homogenized problem \eqref{hs}.
\begin{theorem}[Convergence of wave function]\label{conhs}
Assume that the potential function satisfies \eqref{Vass} and { the initial wave function} $\psi_0^\varepsilon$ is strongly convergent to $\psi_0$ in $H^1(\mathbb{R}^N)$. Let $\hat{T}=\min\{T_0,T^\star_0\}$, $\psi^\varepsilon(t,x)$ and $\psi(t,x)$ be the solutions for \eqref{hs} and \eqref{ehs} respectively. Then
\begin{equation*}
  \lim_{\varepsilon\to0}\|\psi^\varepsilon(t,x)-\psi(t,x)\|_{H^1(\mathbb{R}^N)} = 0 \quad \text{for all }t\in[0,\hat{T\hat{}}).
\end{equation*}
\end{theorem}
\begin{proof}Consider the difference of $\psi^\varepsilon(t,x)-\psi(t,x)$, which satisfies the following equation
\begin{equation}\label{deq}
  i\frac{\partial(\psi^\varepsilon(t,x)-\psi(t,x))}{\partial_t} = -\frac{1}{2}\Delta(\psi^\varepsilon(t,x)-\psi(t,x)) + V(\frac{t}{\varepsilon},x)\psi^\varepsilon(t,x) - V^\star(x)\psi(t,x).
\end{equation}

Multiplying both sides of \eqref{deq} by $-\Delta\overline{(\psi^\varepsilon(t,x)-\psi(t,x))}$, then
\begin{equation}\label{deq1}
  \begin{aligned}
     &ia\left(\frac{\partial(\psi^\varepsilon(t,x)-\psi(t,x))}{\partial_t},\psi^\varepsilon(t,x)-\psi(t,x)\right) \\= &\frac{1}{2}\left(\Delta(\psi^\varepsilon(t,x)-\psi(t,x)),\Delta\overline{(\psi^\varepsilon(t,x)-\psi(t,x))})\right)
- \left(V(\frac{t}{\varepsilon},x)\psi^\varepsilon(t,x),\Delta\overline{(\psi^\varepsilon(t,x)-\psi(t,x))}\right)\\
 &+ V^\star(x)a\left(\psi^\varepsilon(t,x)-\psi(t,x),\psi^\varepsilon(t,x)-\psi(t,x)\right),
  \end{aligned}
\end{equation}
where
\begin{equation*}
  a(\psi(x),\phi(x))=\sum_{i=1}^N\int_{\mathbb{R^N}}\frac{\partial\psi(x)}{\partial x_i}\cdot\frac{\partial\phi(x)}{\partial x_i}\mathrm{d}x
\end{equation*}
for any $\psi(x),\phi(x)\in H^1(\mathbb{R}^N)$.
{ Taking the imaginary }part of \eqref{deq1}, we have
\begin{equation}\label{deq2}
  \begin{aligned}
  \frac{\mathrm{d}}{\mathrm{d}t}\|\psi^\varepsilon(t,x)-\psi(t,x)\|^2_{H^1(\mathbb{R}^N)} \leq \left|\left(V(\frac{t}{\varepsilon},x)-V^\star(x),\psi^\varepsilon(t,x)\Delta(\psi^\varepsilon(t,x)-\psi(t,x))\right)\right|\\
+M\|\psi^\varepsilon(t,x)-\psi(t,x)\|^2_{H^1(\mathbb{R}^N)},
  \end{aligned}
\end{equation}
where $M:=\|V(t,x)\|_{L^\infty}$ and $M$ is finite since the assumption that $V(t,x)\in C_0^\infty(\mathbb{R}\times\mathbb{R}^N;\mathbb{R})$.

Then for any $t\in[0,\hat{T})$, the following inequality can be obtained by applying the Gronwall inequality
\begin{equation}\label{deq3}
  \begin{aligned}
   \|\psi^\varepsilon(t,x)-\psi(t,x)\|^2_{H^1(\mathbb{R}^N)} \leq \left(\|\psi_0^\varepsilon-\psi_0\|^2_{H^1(\mathbb{R}^N)}+ \int_0^te^{(t-s)M}b^\varepsilon(s)\mathrm{d}s\right),
  \end{aligned}
\end{equation}
where the constant $C>0$ independent of $t,x,\varepsilon$ and
\begin{equation*}
  b^\varepsilon(t):=\left|\left(V(\frac{t}{\varepsilon},x)-V^\star(x),\psi^\varepsilon(t,x)\Delta(\psi^\varepsilon(t,x)-\psi(t,x))\right)\right|.
\end{equation*}
For the rapidly oscillating periodic potential,  Theorem \ref{wlr} implies that
\begin{equation*}
  V(\frac{t}{\varepsilon},x)\xrightarrow{\varepsilon\to0}V^\star(x) \quad \text{weakly}-\star~\text{in}~L^\infty(\mathbb{R}^N).
\end{equation*}
Then together with the fact that $\psi^\varepsilon(t,x)\Delta(\psi^\varepsilon(t,x)-\psi(t,x))\in L^1(\mathbb{R}^N)$, we can further deduce that
\begin{equation}\label{deq4}
  b^\varepsilon(t)=\left|\left(V(\frac{t}{\varepsilon},x)-V^\star(x),\psi^\varepsilon(t,x)\Delta(\psi^\varepsilon(t,x)-\psi(t,x))\right)\right|\xrightarrow{\varepsilon\to0}0.
\end{equation}
Moreover the fact that $\psi^\varepsilon(t,x)\Delta(\psi^\varepsilon(t,x)-\psi(t,x))\in L^1(\mathbb{R}^N)$ follows from that
\begin{equation*}
\begin{aligned}
  &\int_{\mathbb{R}^N}\left|\psi^\varepsilon(t,x)\Delta(\psi^\varepsilon(t,x)-\psi(t,x))\right|\mathrm{d}x \\
\leq & \|\psi^\varepsilon(t,x)\|_{H^1(\mathbb{R}^N)}\cdot\|\Delta(\psi^\varepsilon(t,x)-\psi(t,x))\|_{H^{-1}(\mathbb{R}^N)}\\
 <& \infty,
\end{aligned}
\end{equation*}
where the last inequality holds because of Lemma \ref{hsbdd} and \ref{ehsbdd}.

Thus together with \eqref{deq4} and the assumption that $\psi_0^\varepsilon\xrightarrow{\varepsilon\to0}\psi_0$ strongly in $H^1(\mathbb{R}^N)$, we conclude from \eqref{deq3} that
\begin{equation*}
  \|\psi^\varepsilon(t,x)-\psi(t,x)\|^2_{H^1(\mathbb{R}^N)}\xrightarrow{\varepsilon\to0}0 \quad \text{for all }t\in[0,\hat{T}).
\end{equation*}
This proof of Theorem \ref{conhs} is complete.
\end{proof}

The following Theorem is about the convergence of the position density $\rho^\varepsilon$ and the current density $J^\varepsilon$. { This} is a natural consequence of Theorem \ref{conhs}.
\begin{theorem}[Convergence of densities]\label{step1}
Assume that the assumptions in Theorem \ref{conhs} hold. Let $\hat{T}$ be as in Theorem \ref{conhs}, $\psi^\varepsilon(t,x)$ and $\psi(t,x)$ be the solutions for \eqref{hs} and \eqref{ehs} respectively, then
\begin{equation*}
   \rho^\varepsilon\xrightarrow{\varepsilon\to0_+}\rho~ \text{in}~ L^1(\mathbb{R}^N)~ \text{strongly } \text{for all }t\in[0,\hat{T\hat{}}),
\end{equation*}
  and
\begin{equation*}
J^\varepsilon\xrightarrow{\varepsilon\to0_+} J~ \text{in}~ L^1(\mathbb{R}^N)~ \text{strongly } \text{for all }t\in[0,\hat{T\hat{}}).
\end{equation*}
\end{theorem}
\begin{proof}
The uniform boundness of $\psi^\varepsilon$ in Lemma \ref{hsbdd} implies that the corresponding densities $\rho^\varepsilon(t,x),~J^\varepsilon(t,x)\in L^1(\mathbb{R}^N)$, and { so do $\rho(t,x)$ and $J(t,x)$}.

With the definition of $\rho^\varepsilon(t,x)$ and $\rho(t,x)$ in hand, we can show that
\begin{align*}
  &\int_{\mathbb{R}^N}|\rho^\varepsilon(t,x)-\rho(t,x)|\mathrm{d}x \\
   = &\int_{\mathbb{R}^N}\left||\psi^\varepsilon(t,x)|^2-|\psi(t,x)|^2\right|\mathrm{d}x\\
    = &\int_{\mathbb{R}^N}|\psi^\varepsilon(t,x)\overline{\psi^\varepsilon}(t,x)-\psi(t,x)\overline{\psi}(t,x)|\mathrm{d}x\\
   \leq & \int_{\mathbb{R}^N}|\overline{\psi^\varepsilon}(t,x)||\psi^\varepsilon(t,x)-\psi(t,x)|\mathrm{d}x+\int_{\mathbb{R}^N}|\psi(t,x)||\overline{\psi^\varepsilon}(t,x)-\overline{\psi}(t,x)|\mathrm{d}x\\
    \leq& \|\psi^\varepsilon(t,x)\|_{L^2(\mathbb{R}^N)}\|\psi^\varepsilon(t,x)-\psi(t,x)\|_{L^2(\mathbb{R}^N)} + \|\psi(t,x)\|_{L^2(\mathbb{R}^N)}\|\psi^\varepsilon(t,x)-\psi(t,x)\|_{L^2(\mathbb{R}^N)}\\
  &\to0\quad \text{as }\varepsilon\to0,
\end{align*}
where the last assertion follows from the uniform boundness of $\psi^\varepsilon,\psi$ in Lemma \ref{hsbdd} and Lemma \ref{ehsbdd} and convergence of $\psi^\varepsilon$ in Theorem \ref{conhs}.

{ Also the definition of the current densities} $J^\varepsilon(t,x)$ and $J(t,x)$ imply that
\begin{align*}
   & \int_{\mathbb{R}^N}|J^\varepsilon(t,x)-J(t,x)|\mathrm{d}x \\
    = &\int_{\mathbb{R}^N}|\mathrm{Im}(\overline{\psi^\varepsilon}(t,x)\nabla\psi^\varepsilon(t,x))-\mathrm{Im}(\overline{\psi}(t,x)\nabla\psi(t,x))|\mathrm{d}x\\
\leq & \|\psi^\varepsilon(t,x)\|_{L^2(\mathbb{R}^N)}\|\nabla\psi^\varepsilon(t,x)-\nabla\psi(t,x)\|_{L^2(\mathbb{R}^N)} + \|\nabla\psi(t,x)\|_{L^2(\mathbb{R}^N)}\|\psi^\varepsilon(t,x)-\psi(t,x)\|_{L^2(\mathbb{R}^N)}\\
  &\to0\quad \text{as }\varepsilon\to0,
\end{align*}
where the last assertion follows from the uniform boundness of $\psi^\varepsilon,\psi$ in Lemma \ref{hsbdd} and Lemma \ref{ehsbdd} and convergence of $\psi^\varepsilon$ in Theorem \ref{conhs}. The proof of Theorem \ref{step1} is complete.
\end{proof}

\begin{theorem}[Convergence of Bohmian measures]\label{step2}
Assume that the assumptions in Theorem \ref{conhs} hold. Let $\psi^\varepsilon(t,x)$ and $\psi(t,x)$ be the solutions for \eqref{hs} and \eqref{ehs} respectively, then up to extraction of sub-sequences, there exists a limiting measure $\beta\in\mathcal{M}^+(\mathbb{R}_x^N\times\mathbb{R}_p^N)$, such that
\begin{equation*}
  \beta^\varepsilon\to\beta \quad \text{as}\quad \varepsilon\to0_+\quad
\end{equation*}
weakly-$\star$ in $\mathcal{M}^+(\mathbb{R}_x^N\times\mathbb{R}_p^N)$ and the limit $\beta$ is mono-kinetic, i.e.
\begin{equation*}
  \beta(x,p)=\rho(t,x)\delta\left(p-\frac{J(t,x)}{\rho(t,x)}\right).
\end{equation*}
\end{theorem}
\begin{proof}Since $\psi^\varepsilon$ is uniformly bounded in $H^1(\mathbb{R}^N)$ in Lemma \ref{hsbdd}, then by Lemma \ref{wkb1} there exists some measure $\beta$ such that
\begin{equation*}
  \beta^\varepsilon\to\beta \quad \text{as}\quad \varepsilon\to0_+\quad
\end{equation*}
weakly-$\star$ in $\mathcal{M}^+(\mathbb{R}_x^N\times\mathbb{R}_p^N)$.
Then due to the Theorem \ref{wkb2} and the strong convergence of $\rho^\varepsilon(t,x)$ and $J^\varepsilon(t,x)$ stated in Theorem \ref{step1}, we can infer that the limiting measure is is mono-kinetic and given by
\begin{equation*}
  \beta(x,p)=\rho(t,x)\delta\left(p-\frac{J(t,x)}{\rho(t,x)}\right).
\end{equation*}
The proof of Theorem \ref{step2} is complete.
\end{proof}

Now we reformulate the Bohmian mechanics in its Lagrangian formulation(\cite{2014Figalli,2010Markowich}). Let the motion be $P^\varepsilon = u^\varepsilon(x,X^\varepsilon(t,x))$, we can differentiate $P^\varepsilon(t,x)$ $\rho_0^\varepsilon$-a.e. to obtain
\begin{equation}\label{dm}
  \dot{P}^\varepsilon(t,x)=\partial_tu^\varepsilon(x,X^\varepsilon(t,x)) + (u^\varepsilon(x,X^\varepsilon(t,x))\cdot\nabla)u^\varepsilon(x,X^\varepsilon(t,x)).
\end{equation}
With assumptions given on the potential $V(t,x)$, we can infer that the wave function $\psi^\varepsilon(t,x)$ is smooth enough(\cite{1997Gasser}). Then one can deduce a closed system of equations for the densities $\rho^\varepsilon$ and $J^\varepsilon$ from the oscillating Schr\"{o}dinger equations \eqref{hs}, i.e. the well-known {  hydrodynamic formulation of quantum mechanics} (\cite{2006Wyatt}) which holds in the sense of distribution
\begin{equation} \label{hcs}
\begin{cases}
\partial_t\rho^\varepsilon + \mathrm{div}J^\varepsilon =0,\\
\partial_t J^\varepsilon + \mathrm{div}\left(\frac{J^\varepsilon\otimes J^\varepsilon}{\rho^\varepsilon}\right) + \rho^\varepsilon\nabla V = \frac{1}{2}\rho^\varepsilon\nabla\left(\frac{\Delta\sqrt{\rho^\varepsilon}}{\sqrt{\rho^\varepsilon}}\right).
\end{cases}
\end{equation}
Note that $J^\varepsilon=\rho^\varepsilon u^\varepsilon$, together with \eqref{hcs} we can formulate the following system of ordinary differential equations which fully determines the quantum mechanical dynamics
\begin{equation} \label{thcs}
\begin{cases}
\dot{X}^\varepsilon = P^\varepsilon,\\
\dot{P}^\varepsilon = -\nabla V(\frac{t}{\varepsilon},X^\varepsilon) + \frac{1}{2}\frac{\Delta \sqrt{\rho^\varepsilon}}{\sqrt{\rho^\varepsilon}},\\
X^\varepsilon(t,x)|_{t=0} =x, \quad  \\
P^\varepsilon(t,x)|_{t=0} = u^\varepsilon(0,x),
\end{cases}
\end{equation}
where the initial velocity given by
 $$u^\varepsilon(0,x):=\frac{J_0^\varepsilon}{\rho_0^\varepsilon}=\frac{\mathrm{Im(\psi^\varepsilon_0\nabla\psi^\varepsilon_0)}}{|\psi^\varepsilon_0|^2}.$$

Theorem \ref{Ys} implies that the limit of Bohmian trajectories is connected to the Young measure generated by the Bohmian dynamics. To proceed further we recall the following definition of the Young measure. For the solutions $\Phi^\varepsilon(t,y):= \left(X^\varepsilon(t,y),P^\varepsilon(t,y)\right)$ which is measurable in $t$ and $y$, there exists an associated Young measure $(t,y)\to\Upsilon_{t,y}(dx,dp)$ (\cite{1988Ball,1997N,2014Figalli})
\begin{equation*}
  \Upsilon_{t,y}:\mathbb{R}_t\times\mathbb{R}_y^N\to\mathcal{M}^+(\mathbb{R}_y^N\times\mathbb{R}_p^N),
\end{equation*}
where the Young measure is defined through the following limit: for any test function $\varphi\in L^1(\mathbb{R}_t\times\mathbb{R}_y^N;C_0(\mathbb{R}^{2N})),$
\begin{equation*}
  \lim_{\varepsilon\to0}\underset{{\mathbb{R}\times\mathbb{R}^N}}\iint\varphi(t,y,\Phi^\varepsilon(t,y))\mathrm{d}y\mathrm{d}t=
\underset{{\mathbb{R}\times\mathbb{R}^N}}\iint\underset{{\mathbb{R}^{2N}}}\iint\varphi(t,y,x,p)\Upsilon_{t,y}(dx,dp)\mathrm{d}y\mathrm{d}t.
\end{equation*}

{ Let us draw the conclusion on the convergence of Bohmian} trajectories \eqref{thcs} through the connection between the limiting Bohmian measure $\beta$ and the Young measure $\Upsilon_{t,y}(t,y)$. From \cite{1997N}, the Bohmian flow $x\to X^\varepsilon(t,x)$ is well-defined $\rho_0^\varepsilon$-a.e. for all $t\in\mathbb{R}^+$ regardless of the continuity of $u^\varepsilon$ and $\rho_0^\varepsilon:=|\psi^\varepsilon_0|^2$. The assertion also holds for the Bohmian flow $x\to X(t,x)$ with measure $\rho_0:=|\psi_0|^2$. Note that $u(t,x)$ { is not Lipschitz continuous} for all $t\in\mathbb{R}^+$, there exists some constant $T^\star>0$ ($T^\star$ could be very small) such that the Bohmian flow $X(t,x)$ is one-to-one for $t\in[0,T^\star)$ (\cite{1995H}).


\begin{theorem}[Convergence of Bohmian trajectories]\label{step3}Under the assumptions that the potential function { satisfies} \eqref{Vass} and the initial wave functions $\psi_0^\varepsilon$ is strongly convergence in $H^1(\mathbb{R}^N)$, then for all compact time intervals $I\subset[0,T)$ and $T=\min\{T_0,T^\star_0,T^\star\}$, the following convergence of corresponding Bohmian trajectories and momentum of \eqref{thcs} hold locally in measure $\rho_0$ on $\{I\times\mathrm{supp}\rho_0\}\subset \mathbb{R}\times\mathbb{R}^N$
\begin{equation*}
  X^\varepsilon\xrightarrow{\varepsilon\to0_+} X, \quad  P^\varepsilon\xrightarrow{\varepsilon\to0_+}P.
\end{equation*}
Furthermore, the limit position-momentum pair $(X,P)$ is the solution of
\begin{equation} \label{ethcs}
\begin{cases}
\dot{X} = P,\\
\dot{P} = -\nabla V^\star(X) + \frac{1}{2}\frac{\Delta \sqrt{\rho}}{\sqrt{\rho}},\\
P(t,x) = u(t,X(t,x)),\\
X(t,x)|_{t=0} =x \quad \text{and}\quad  P(t,x)|_{t=0} = u(0,x).
\end{cases}
\end{equation}
More precisely, for every $\delta>0$ and every Borel set $\Omega\subset\{I\times\mathrm{supp}\rho_0\}$ with finite Lebesgue measure $\mathscr{L}^{N+1}$, it holds that
\begin{equation*}
  \lim_{\varepsilon\to0}\mathscr{L}^{N+1}\left(\{(t,y)\in\Omega:|(X^\varepsilon(t,y),P^\varepsilon(t,y))-(X(t,y),P(t,y))|\geq\delta\}\right)=0.
\end{equation*}
\end{theorem}
\begin{proof}
The assumption that $\psi_0^\varepsilon\xrightarrow{\varepsilon\to0}\psi_0$ strongly in $H^1(\mathbb{R}^N)$ implies that
\begin{equation}\label{ass}
  \rho_0^\varepsilon\xrightarrow{\varepsilon\to0_+}\rho_0\quad \text{strongly in } L^1_+(\mathbb{R}^N).
\end{equation}
Then the following assertion can be obtained from \cite{2010Markowich} together with \eqref{ass}
\begin{equation}\label{ass1}
  \beta(t,x,p)=\underset{\mathbb{R}_y^N}\int\Upsilon_{t,y}(t,p)\rho_0(y)\mathrm{d}y.
\end{equation}

With identity \eqref{ass1} and the convergence of Bohmian measures (Theorem \ref{step2}) in mind, we can infer that the Young measure $\Upsilon_{t,y}(t,p)$ is supported in a single point by analogous reasonings as in \cite{2014Figalli}. That is the Young measure $\Upsilon_{t,y}(t,p)$ can be given by
\begin{equation*}
  \Upsilon_{t,y}(t,p) = \delta(x-X(x,y))\delta(p-u(t,X(t,y))) \quad \text{a.e. on }\mathrm{supp}\rho_0\subset\mathbb{R}^N.
\end{equation*}
Thus the following assertion holds from Theorem \ref{Ys} in measure theory
\begin{equation*}
  X^\varepsilon\xrightarrow{\varepsilon\to0_+} X, \quad  P^\varepsilon\xrightarrow{\varepsilon\to0_+}P,
\end{equation*}
locally in measure on $\{I\times\mathrm{supp}\rho_0\}\subset \mathbb{R}\times\mathbb{R}^N$. The proof of Theorem \ref{step3} is complete.
\end{proof}

\section{Conclusion and Discussion}
\noindent

In this paper, { we have studied the limit of the Bohmian trajectory under} assumptions on the potential and on the considered class of initial wave function $\psi_0^\varepsilon$. { The Bohmian trajectories we considered are guided} by a family of wave functions which satisfied the Schr\"{o}dinger equations with rapidly oscillating potential. Due to the complicated relation between wave functions and
the Bohmian trajectories, the convergence of the  Schr\"{o}dinger equations does not directly imply convergence of the corresponding Bohmian trajectories. In order to gain more insight, we { recalled} the Bohmian measure which naturally arise in the Bohmian interpretation of quantum mechanics and { studied the   limit} of these measures. On this foundation, together with a well-known theorem in measure theory for Young measure corresponding to the Bohmian flow, { we managed to deduce} that the Bohmian trajectories converge locally in measure. However, the conclusion holds when the Bohmian trajectory is one-to-one on finite time interval as a result of the lack of continuity of velocity field. { We also have several comments in the following.}
{
\begin{itemize}
  \item [1.]The study in this paper is important to understand how the time-oscillating Schr\"{o}dinger equations affect the corresponding Bohmian trajectories in some finite time interval. Our approach can also be applied to a class of Schr\"{o}dinger equations when the potential has other heterogeneities (\cite{1978Lions,1999Donato}). With { techniques} provided in homogenization (\cite{2014duan}), we can study the Bohmian trajectory guided by the effective Schr\"{o}dinger equation. This can be taken as a simplified method in the study of Bohmian mechanics.
  \item [2.]We only studied the Schr\"{o}dinger equation with linear potential. It would be interesting and useful to consider the Bohmian trajectories guided by the Schr\"{o}dinger equation with nonlinear potential, e.g. $V(\cdot)=g(|\psi|^2)$. The main obstacle at this point is the existence of Bohmian trajectories. It is conceivable that the results in this paper could be generalized with  suitable assumptions on the initial wave function and the nonlinear potential.
  \item [3.]Chaos in quantum mechanics is a problem of great current interest. Bohmian trajectories  that approach the nodal point are chaotic in general, and the Bohmian trajectories could be chaotic if the classical potential is { time dependent in more than one-dimensional system} (\cite{2008Contopoulos,2020Contopoulos,2006chaos,1996Iacomelli,2005Wisniacki,2006Wisniacki}). In this paper, the convergence of Bohmian trajectory is valid in a finite time interval where the Bohmian flow is one-to-one and chaos has not occurred in this time interval.  However, when the Bohmian trajectory is chaotic, whether the Bohmian trajectory guided by the time-oscillating Schr\"{o}dinger equation has a limit, what is the expression of the limit and  how the convergence happens are still unknown and worthy of study.
\end{itemize}

}

\section*{Acknowledgement}

 We would like to thank Stephen Wiggins for bringing us to the study of Bohmian mechanics.

\section*{Data Availability}
The data that support the findings of this study are available from the corresponding author
upon reasonable request.

\end{document}